\documentclass[a4paper,UKenglish,cleveref, autoref]{amsart}

\usepackage{amsaddr}
\usepackage{geometry} 
\usepackage[ruled,linesnumbered]{algorithm2e}
\usepackage{amsmath,amssymb,amsthm}
\usepackage{graphicx}

\newtheorem{theorem}{Theorem}
\newtheorem{property}{Property}

\newcommand{\CDAWG}{\ensuremath{\mathsf{CDAWG}}}
\newcommand{\ST}{\ensuremath{\mathsf{ST}}}

\newcommand{\MS}{\ensuremath{\mathsf{MS}}}

\newcommand{\BWT}{\ensuremath{\mathsf{BWT}}}
\newcommand{\RLBWT}{\ensuremath{\mathsf{RLBWT}}}
\newcommand{\REV}[1]{\ensuremath{\overline{#1}}}

\newcommand{\SLT}{\ensuremath{\mathsf{SLT}}}

\newcommand{\DESCRIPTOR}{\ensuremath{\mathtt{id}}}
\newcommand{\ExtendRight}{\ensuremath{\mathtt{extendRight}}}
\newcommand{\ContractRight}{\ensuremath{\mathtt{contractRight}}}
\newcommand{\ExtendLeft}{\ensuremath{\mathtt{extendLeft}}}
\newcommand{\ContractLeft}{\ensuremath{\mathtt{contractLeft}}}

\title{Fully-functional bidirectional Burrows-Wheeler indexes}

\author{Djamal Belazzougui}
\address{CAPA, DTISI, Centre de Recherche sur l'Information Scientifique et Technique \\ Algiers, Algeria.}
\email{dbelazzougui@cerist.dz}

\author{Fabio Cunial}
\address{Max Planck Institute for Molecular Cell Biology and Genetics (MPI-CBG) \\ Center for Systems Biology Dresden (CSBD) \\ Dresden, Germany.}
\email{cunial@mpi-cbg.de}

\begin{document}

\maketitle

\begin{abstract}
Given a string $T$ on an alphabet of size $\sigma$, we describe a bidirectional Burrows-Wheeler index that takes $O(|T|\log{\sigma})$ bits of space, and that supports the addition \emph{and removal} of one character, on the left or right side of any substring of $T$, in constant time. Previously known data structures that used the same space allowed constant-time addition to any substring of $T$, but they could support removal only from specific substrings of $T$. We also describe an index that supports bidirectional addition and removal in $O(\log{\log{|T|}})$ time, and that takes a number of words proportional to the number of left and right extensions of the maximal repeats of $T$. We use such fully-functional indexes to implement bidirectional, frequency-aware, variable-order de Bruijn graphs with no upper bound on their order, and supporting natural criteria for increasing and decreasing the order during traversal.
\end{abstract}

\section{Introduction}

A \emph{bidirectional index} on a string $T$ is a data structure that represents any substring $W$ of $T$ as a constant-size descriptor that recapitulates the set of all starting positions of $W$ in $T$, and the set of all ending positions of $W$ in $T$. Such a representation allows extending $W$ with a character in both directions, enumerating the distinct characters that occur after $W$ in both directions, and switching direction during extension. All existing bidirectional indexes can be seen as updating positions in the suffix tree of $T$ and in the suffix tree of the reverse of $T$, either literally, as in the \emph{affix tree} \cite{maass2000linear,stoye2000affix}, or in compact representations, like the \emph{affix array} \cite{strothmann2007affix} and the \emph{bidirectional Burrows-Wheeler transform} (BWT) \cite{schnattinger2012bidirectional}. \emph{Synchronous} bidirectional indexes maintain a position in both trees at every extension step, whereas \emph{asynchronous} indexes maintain a position in just one tree, and compute the position in the other only when the user needs to change direction \cite{canovas2017full}. Applications of bidirectional indexes to bioinformatics, like read mapping with mismatches and searching for RNA secondary structures, have used until now the ability of bidirectional indexes to \emph{add} characters both to the left and to the right of a string (an operation called \emph{extension}: see e.g. \cite{gog2014multi,lam2009high,mauri2003pattern,russo2009approximate,schnattinger2012bidirectional,strothmann2007affix} for a small sampler), whereas \emph{removing} characters from the left and from the right (called \emph{contraction}) has only been conjectured to be useful \cite{belazzougui2016bidirectional,canovas2017full}, and it has been supported efficiently just for right-maximal and left-maximal substrings of $T$, respectively (defined in Section~\ref{sec:preliminaries}), or for strings that occur just once in $T$, for which the implementation is straightforward (see e.g. \cite{BCKM13,munro2017space}).
%

In this paper we describe a simple method for removing characters from the left or from the right of any substring of $T$, 
based just on the ability to measure the length of the \emph{maximal repeats} of $T$ (defined in Section~\ref{sec:preliminaries}). Using the recent observation that all such lengths can be represented in $O(|T|)$ bits of space \cite{belazzougui2017framework}, we show that bidirectional contraction can be supported in constant time with the bidirectional BWT index described in~\cite{BCKM13}, within the same space budget and without changing the complexity of its construction. Our contraction algorithm can also be implemented on top of an existing representation of the suffix tree, based on the \emph{Compact Directed Acyclic Word Graph} (CDAWG), that takes a number of words proportional just to the number of left and right extensions of the maximal repeats of $T$~\cite{belazzougui2017representing}: this yields an index that supports, in the same asymptotic space, bidirectional extension and contraction of any substring of $T$ in $O(\log{\log{|T|}})$ time.

Having both bidirectional extension and contraction enables several applications, among which a de Bruijn graph that stores the frequency of its $k$-mers, allows for bidirectional navigation, and supports any value of $k$, as well as increasing and decreasing the value of $k$, \emph{with no limit on the maximum $k$ allowed}. We call such a data structure an \emph{infinite-order de Bruijn graph}, and we describe an implementation that takes $O(|T|\log{\sigma})$ bits of space (where $\sigma$ is the size of the alphabet), and that supports all operations in constant time, as well as another implementation that takes a number of words proportional to the left an right extensions of the maximal repeats of $T$, and that supports all operations in $O(\log{\log{|T|}})$ time. The latter representation establishes a connection between de Bruijn graphs and CDAWGs that was not known before. Our query times are comparable to those of the variable-order, bidirectional representation described in \cite{belazzougui2016bidirectional}, which supports navigation and changing order in $O(\log{K})$ time (assuming constant~$\sigma$), but is frequency-oblivious and requires a maximum order $K$ to be specified during construction. This competitor has the advantage of taking just $O(m\log{K})$ bits of space, where $m$ is the number of distinct $K$-mers, 
and of allowing the user to specify by how much the order should be changed in each query (the changes in order supported by our index are detailed in Sections \ref{sec:contract} and \ref{sec:dbgCDAWG}). 
The variable-order representation described in \cite{diaz2018assembling} takes constant time (assuming constant $\sigma$) 
to implement changes in order that are similar to those supported by our index, and uses just $O(m)$ bits of space; 
however, it is unidirectional, frequency-oblivious, and it requires again a maximum $K$ to be known at construction time.

\begin{figure}[h!]
\centering
\includegraphics[width=0.9\textwidth]{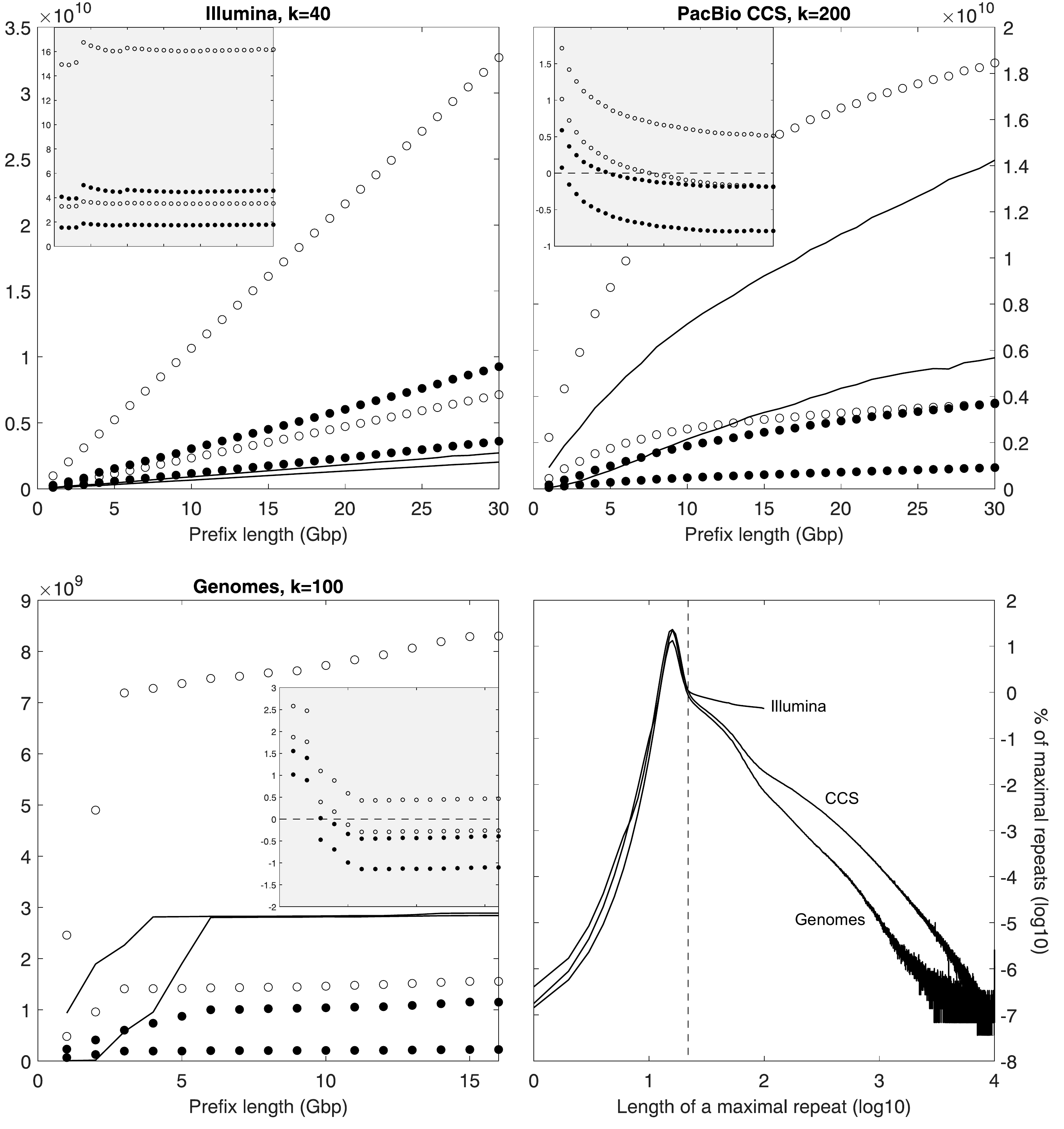}
\caption{
Number of $k$-mers and repeated $k$-mers (lines), maximal repeats and bidirectional extensions of maximal repeats (white circles), and maximal repeats of length at least 20 and the bidirectional extensions connecting them 
(black circles), for prefixes of a human read dataset produced with Illumina and PacBio CCS technologies. Bottom left: prefixes of the concatenation of 5 human genome assemblies. Bottom right: fraction of maximal repeats of each length in the three datasets (the vertical line is at length 20). Inserts show the number of maximal repeats and extensions, divided by the number of repeated $k$-mers (in $\log_{10}$ scale for CCS and genomes). Decreasing $k$ down to 20 (Illumina), 50 (CCS) and 25 (genomes) yields similar plots. $k$-mers are counted with KMC 3 \cite{kokot2017kmc}, and are considered distinct from their reverse complements. 
%
}
\label{fig:growth}
\end{figure}

We conjecture that a de Bruijn graph representation based on the CDAWG might be useful for assembling the recently introduced PacBio CCS reads, which have the same 2\% error rate as  Illumina short reads but an average length of 15 kilobases (see e.g. \cite{Wenger519025}). Such read sets contain long exact repeats, of length up to ten thousand, so it might be desirable to set $k$ to large values and to decrease it dynamically, down to a minimum value $\tau$. Moreover, most maximal repeats are short (Figure~\ref{fig:growth}, bottom right), and we can remove from the CDAWG all maximal repeats shorter than $\tau$, and all arcs adjacent to them, while still being able to represent all de Bruijn graphs of order at least $\tau$ (see Section \ref{sec:dbgCDAWG}). For practical values of $k$, the number of nodes and arcs in such a pruned CDAWG grows more slowly than the number of distinct $k$-mers (Figure~\ref{fig:growth}, top right; reads from the Genome in a Bottle consortium\footnote{\texttt{ftp://ftp-trace.ncbi.nlm.nih.gov/giab/ftp/data/AshkenazimTrio/HG002\_NA24385\_son/ \\ PacBio\_CCS\_15kb/}}), suggesting that our data structure might be competitive in space with the state of the art, whose size is proportional to the number of $k$-mers for a specific value of $k$. The same observation applies to repetitive datasets: for example, the de Bruijn graph of a set of individuals from the same species has applications in population genomics, and the de Bruijn graph of a set of genomes from related species is used in comparative genomics  \cite{minkin2019scalable,minkin2016twopaco}. In Figure~\ref{fig:growth}, bottom left, we experiment with the concatenation of assemblies \texttt{hg16}, \texttt{hg17}, \texttt{hg18}, \texttt{hg19} and \texttt{hg38} of the human genome from the UCSC Genome Browser\footnote{\texttt{http://hgdownload.soe.ucsc.edu/downloads.html\#human}} (a benchmark dataset from \cite{baier2015graphical,minkin2016twopaco}), and we observe exact repeats of length up to 489 million.
Our data structure might also be useful with noisy long reads after error correction. 
Even in short-read Illumina datasets, the number of maximal repeats and of their extensions after pruning is just a small multiple of the number of distinct $k$-mers (Figure~\ref{fig:growth}, top left; reads from the Illumina Platinum project\footnote{\texttt{https://www.ebi.ac.uk/ena/data/view/PRJEB3381}, run \texttt{ERR194146}, file \texttt{ERR194146\_1.fastq.gz}, read length 101.}).

Finally, recall that our de Bruijn graph representations allow access to the frequency of a node or arc: this might be useful for avoiding repetitive regions during assembly, or for reconstructing only those \cite{koch2014repark}, for assembling metagenomes with non-uniform sequencing depths \cite{li2016megahit}, or for inferring transcripts with different expression levels \cite{pandey2017debgr}.


\section{Preliminaries} \label{sec:preliminaries}

\subsection{Strings}

Let $\Sigma=[1..\sigma]$ be an integer alphabet, let $\#=0$ be a separator not in $\Sigma$, and let $T=[1..\sigma]^{n-1}$ be a string. We denote by $\REV{W}$ the reverse of a string $W$, i.e. string $W$ written from right to left, and we call $W$ a \emph{k}-mer iff $|W|=k$. We denote by $f_{T}(W)$ the number of (possibly overlapping) occurrences of a string $W$ in the circular version of $T$. A \emph{repeat} $W$ is a string that satisfies $f_{T}(W)>1$. We denote by $\Sigma^{\ell}_{T}(W)$ the set of \emph{left-extensions of $W$}, i.e. the set of characters $\{a \in [0..\sigma] : f_{T}(aW)>0\}$. Symmetrically, we denote by $\Sigma^{r}_{T}(W)$ the set of \emph{right-extensions of $W$}, i.e. the set of characters $\{b \in [0..\sigma] : f_{T}(Wb)>0\}$. A repeat $W$ is \emph{right-maximal} (respectively, \emph{left-maximal}) iff $|\Sigma^{r}_{T}(W)|>1$ (respectively, iff $|\Sigma^{\ell}_{T}(W)|>1$). It is well-known that $T$ can have at most $n-1$ right-maximal substrings and at most $n-1$ left-maximal substrings. A \emph{maximal repeat} of $T$ (called \emph{balanced substring} in \cite{strothmann2007affix}) is a repeat that is both left- and right-maximal.

The \emph{unidirectional de Bruijn graph} of order $k$ of $T$ is a directed graph $(V,E)$ whose node set $V$ is in one-to-one correspondence with the set of distinct $k$-mers that occur in $T$; there is an arc $(v,w) \in E$ for every distinct $(k+1)$-mer $W$ such that both $W[1..k]$ and $W[2..k+1]$ occur in $T$, and such arc is labelled with character $W[k+1]$. In some formulations, $E$ contains just those arcs that correspond to $(k+1)$-mers that occur in $T$: in this case, a $k$-mer is right-maximal (respectively, left-maximal) in $T$ iff its corresponding node in $V$ has at at least two outgoing (respectively, incoming) arcs. The \emph{bidirectional} de Bruijn graph is defined symmetrically.

We denote by $\ST_T$ the \emph{suffix tree} of $T\#$, and by $\REV{\ST}_{T}$ the suffix tree of $\REV{T}\#$. We assume the reader to be already familiar with the basics of suffix trees, including suffix links, which we do not further describe here. We denote by $\ell(v)$ the label of a node $v$ of a suffix tree, and we say that $v$ is the \emph{locus} of all substrings $W[1..k]$ of $T$ where $|\ell(u)| < k \leq |\ell(v)|$, $u$ is the parent of $v$, and $W=\ell(v)$. It is well-known that a substring $W$ of $T$ is right-maximal (respectively, left-maximal) iff $W=\ell(v)$ for some internal node $v$ of $\ST_T$ (respectively, for some internal node $v$ of $\REV{\ST}_T$). Suffix links and internal nodes of $\ST_T$ form a tree, called the \emph{suffix-link tree} of $T$ and denoted by $\SLT_T$, and inverting the direction of all suffix links yields the so-called \emph{explicit Weiner links}. Given an internal node $v$ and a character $a \in [0..\sigma]$, it might happen that string $a\ell(v)$ occurs in $T$ but is not right-maximal, i.e. it is not the label of any internal node of $\ST_T$: all such left extensions of internal nodes that end in the middle of an edge are called \emph{implicit Weiner links}. An internal node $v$ of $\ST_T$ can have more than one outgoing Weiner link, and all such Weiner links have distinct labels: in this case, $\ell(v)$ is a maximal repeat, as well as the label of a node in $\REV{\ST}_T$. Maximal repeats and implicit Weiner links are related by the following simple property, which was already hinted at in \cite{apostolico2000optimal}:

\begin{property} \label{property:maxRepeats}
Let $v$ be an internal node of $\ST_T$. If there is an implicit Weiner link from $v$, then $\ell(v)$ is a maximal repeat of $T$.
\end{property}

It is known that the number of suffix links (or, equivalently, of explicit Weiner links) is upper-bounded by $2n-2$, and that the number of implicit Weiner links can be upper-bounded by $2n-2$ as well. We call $\SLT^{*}_{T}$ a version of $\SLT_T$ augmented with implicit Weiner links and with nodes corresponding to their destinations. We say that a maximal repeat $W$ of $T$ is \emph{rightmost} if no string $WV$ with $V \in [0..\sigma]^+$ is left-maximal in $T$. Symmetrically, we say that a maximal repeat $W$ of $T$ is \emph{leftmost} if no string $VW$ with $V \in [0..\sigma]^+$ is right-maximal in $T$. Since left-maximality is closed under prefix operation, it is easy to see that the maximal repeats of $T$ are all and only the nodes of $\ST_T$ that lie on paths that start from the root and that end at nodes labelled by rightmost maximal repeats. We call this the \emph{maximal repeat subgraph of $\ST_T$} (Figure \ref{fig:contract}b). Clearly the maximal repeats of $T$ coincide with the branching nodes of $\REV{\SLT}^{*}_T$ (Figure \ref{fig:contract}a), and the rightmost maximal repeats of $T$ coincide with the leaves of $\REV{\SLT}_T$. Thus, it is easy to see that $\REV{\SLT}_T$ (a trie) is a subdivision of the maximal repeat subgraph of $\ST_T$ (a compact trie), and that the nodes in the unary paths of $\REV{\SLT}_T$ are in one-to-one correspondence with the internal nodes of $\REV{\ST}_T$ that are not maximal repeats (see Figures \ref{fig:contract}a and \ref{fig:contract}b for an example, and see Section 2.1 in \cite{belazzougui2017framework} for an extended explanation). 
The following property is thus immediate (and symmetrical notions hold for $\REV{\ST}_T$, $\SLT^{*}_T$, and leftmost maximal repeats):

\begin{property} \label{property:affixLinks}
Let $v$ be an internal node of $\ST_T$. The locus $w$ of $\REV{\ell(v)}$ in $\REV{\ST}_T$ is such that $\ell(w)$ is the reverse of a maximal repeat of $T$.
\end{property}

The \emph{compact directed acyclic word graph} of a string $T$ (denoted by $\CDAWG_T$ in what follows) is the minimal compact automaton that recognizes the suffixes of $T$ \cite{blumer1987complete,CrochemoreV97}. We denote by $\REV{\CDAWG}_T$ the CDAWG of the reverse of $T$, by $e_T$ the number of arcs in $\CDAWG_T$, and by $\REV{e}_T$ the number of arcs in $\REV{\CDAWG}_T$. The CDAWG of $T$ can be seen as the minimization of $\ST_T$, in which all leaves are merged to the same node (the sink, that represents $T$ itself), and in which all nodes except the sink are in one-to-one correspondence with the maximal repeats of $T$ \cite{Raffinot2001}. Every arc of $\CDAWG_T$ is labeled by a substring of $T$, and the out-neighbors $w_1,\dots,w_k$ of every node $v$ of $\CDAWG_T$ are sorted according to the lexicographic order of the distinct labels of arcs $(v,w_1),\dots,(v,w_k)$. Since there is a bijection between the nodes of $\CDAWG_T$ and the maximal repeats of $T$, the node $v'$ of $\CDAWG_T$ with $\ell(v')=W$ is the equivalence class of the nodes $\{v_1,\dots,v_k\}$ of $\ST_T$ such that $\ell(v_i)=W[i..|W|]$ for all $i \in [1..k]$, and such that $v_k,v_{k-1},\dots,v_1$ is a maximal unary path of explicit Weiner links. The subtrees of $\ST_T$ rooted at all such nodes are isomorphic. It follows that a right-maximal string can be identified by the maximal repeat $W$ it belongs to, and by the length of the corresponding suffix of $W$ (see \cite{belazzougui2017representing} for an extended explanation). 

We assume the reader to be familiar with the Burrows-Wheeler transform of $T$, which we denote by $\BWT_T$ (we use $\REV{\BWT}_T$ to denote the BWT of the reverse of $T$) and we don't further describe here. We say that $\BWT_{T}[i..j]$ is a \emph{run} iff: (1) $\BWT_{T}[k]=c \in [0..\sigma]$ for all $k \in [i..j]$; (2) every substring $\BWT_{T}[i'..j']$ such that $i' \leq i$, $j' \geq j$, and $[i'..j'] \neq [i..j]$, contains at least two distinct characters. We denote by $\mathcal{R}_{T}$ the set of all triplets $(c,i,j)$ such that $\BWT_{T}[i..j]$ is a run of character $c$, and we use $\REV{\mathcal{R}}_{T}$ to denote the set of runs of $\REV{\BWT}_T$. It is known that $|\mathcal{R}_T|$ is at most equal to the number of arcs in $\CDAWG_T$ \cite{belazzougui2015composite}.

Given a second string $S \in [1..\sigma]^+$, the \emph{matching statistics array} $\MS_{S,T}$ of $S$ with respect to $T$ is an array of length $|S|$ such that $\MS_{S,T}[i]$ is the largest $j$ such that $S[i..i+j-1]$ occurs in $T$.

In the rest of the paper we drop subscripts whenever they are clear from the context.

\begin{figure}[t]
\centering
\includegraphics[width=1\textwidth]{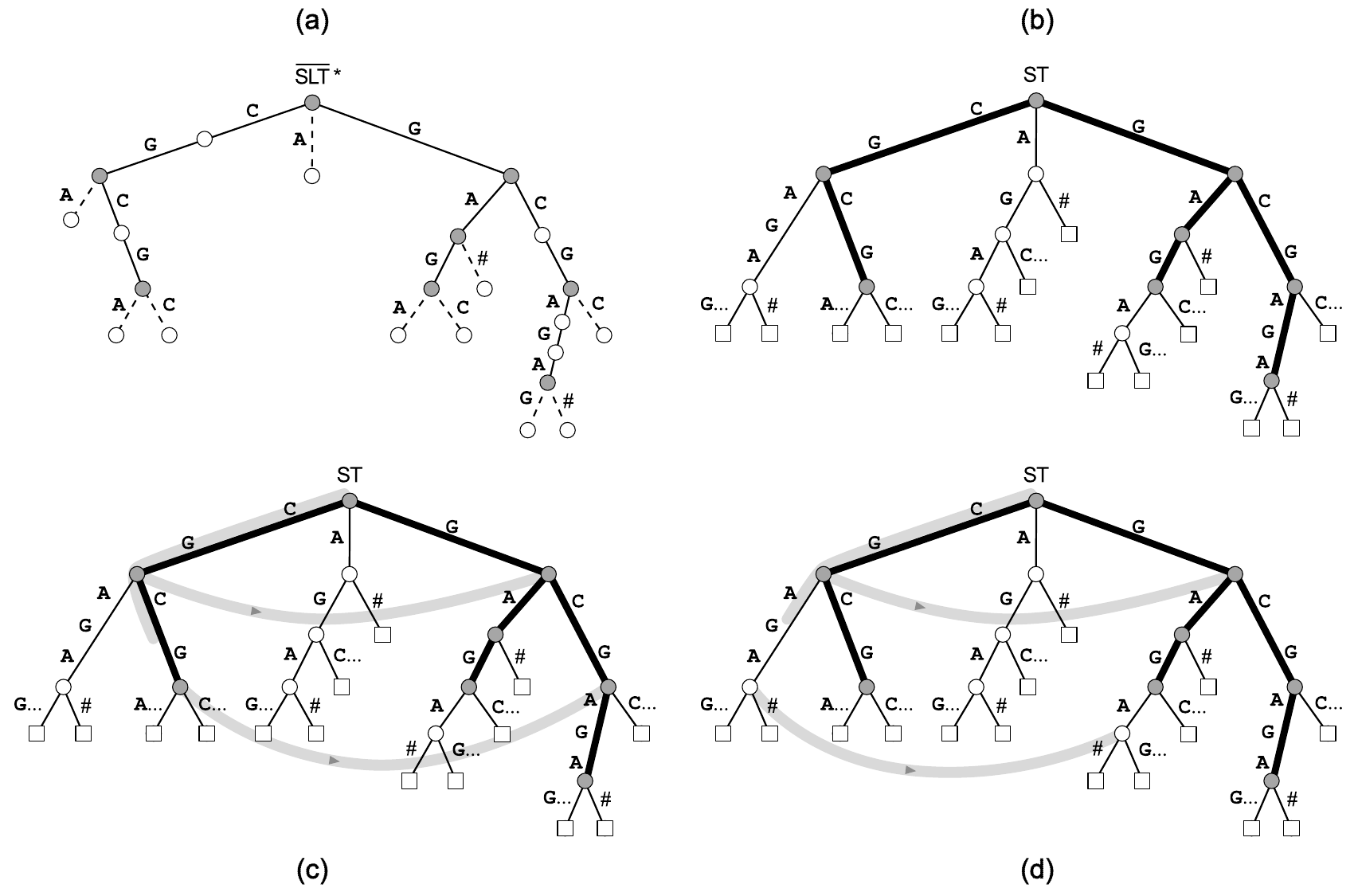}
\caption{
Left-contraction of a substrings that is not right-maximal. (a) The extended suffix-link tree $\REV{\SLT}^{*}_T$ of string $T=\mathtt{CGCGCGAGAGCGAGA\#}$. Nodes that correspond to maximal repeats are highlighted in grey. Implicit Weiner links are dashed. (b) $\ST_T$ (thin lines) with $\REV{\SLT}_T$ overlaid (thick lines). Nodes that correspond to maximal repeats are in grey. Labels of edges to leaves are shortened. (c) Left-contraction of substring $aW=\mathtt{CGC}$. The edge to which $aW$ belongs is projected to another edge by suffix links (thick grey lines). (d) Left-contraction of substring $aW=\mathtt{CGA}$. The edge to which $aW$ belongs is projected to a path by suffix links.
}
\label{fig:contract}
\end{figure}

\subsection{String indexes} \label{sec:stringIndexes}

A \emph{bidirectional index} is a data structure that, given a constant-space descriptor $\mathtt{id}(W)$ of a substring $W$ of $T$, supports the following operations: $\mathtt{extendRight}(\mathtt{id}(W),a)=\mathtt{id}(Wa)$ if $f(Wa)>0$, or an error otherwise; $\mathtt{enumerateRight}(\mathtt{id}(W))=\{\mathtt{id}(Wa) : a \in \Sigma, f(Wa)>0 \}$; $\mathtt{isRightMaximal}(\mathtt{id}(W))=\mathtt{true}$ iff $|\mathtt{enumerateRight}(\mathtt{id}(W))|>1$. Operations $\mathtt{extendLeft}$, $\mathtt{enumerateLeft}$ and $\mathtt{isLeftMaximal}$ are defined symmetrically. We consider bidirectional indexes based on the BWT: specifically, we denote with $\mathbb{I}(W,T)$ the function that maps a substring $W$ of $T$ to the interval of $W$ in $\BWT$, i.e. to the interval of all suffixes of $T\#$ that start with $W$, and we use $\DESCRIPTOR(W)=(\mathbb{I}(W,T),\mathbb{I}(\REV{W},\REV{T}),|W|)$ as a constant-space descriptor of $W$. A number of bidirectional BWT indexes have been described in the literature; in this paper we are just interested in the data structure from \cite{BCKM13}, which supports all operations in linear time in the size of their output, takes $O(|T|\log{\sigma})$ bits of space, and can be built in randomized $O(|T|)$ time and $O(|T|\log{\sigma})$ bits of working space.

Given a string $T \in [1..\sigma]^{n-1}\#$, we call \emph{run-length encoded BWT} ($\RLBWT_T$) any representation of $\BWT_T$ that takes $O(|\mathcal{R}_T|)$ words of space and supports the well-known rank and select operations (see e.g. \cite{makinen2005succinct1,MakinenNSV10,SirenVMN08}). It is easy to implement a version of $\RLBWT_T$ that supports rank and select in $O(\log{\log{n}})$ time~\cite{belazzougui2015composite}. In this paper we use the representation of the suffix tree based on the CDAWG described in \cite{belazzougui2017representing}, which takes just $O(e+\REV{e})$ words of space by augmenting $\CDAWG$ and $\REV{\CDAWG}$ with the RLBWT of $T$ and $\REV{T}$. Such a data structure describes a node $v$ of $\ST$ as a tuple $\mathtt{id}(v)=(v',|\ell(v)|,i,j)$, where $v'$ is the node in $\CDAWG$ that corresponds to the equivalence class of $v$, and $[i..j]$ is the interval of $\ell(v)$ in $\BWT$. For every node $v$ of $\CDAWG$, the index stores, among other things: $|\ell(v)|$ in a variable $v\mathtt{.length}$; the number $v\mathtt{.size}$ of right-maximal strings that belong to its equivalence class; and the interval $[v\mathtt{.first}..v\mathtt{.last}]$ of $\ell(v)$ in $\BWT_T$. For every arc $\gamma=(v,w)$ of $\CDAWG$, the index stores the first character of $\ell(\gamma)$ in a variable $\gamma\mathtt{.char}$, and the number of characters of the right extension implied by $\gamma$ in a variable $\gamma\mathtt{.right}$. Finally, we add to the CDAWG all arcs $(v,w,c)$ such that $w$ is the equivalence class of the destination of a Weiner link from $v$ labeled by character $c$ in $\ST_T$, as well as the reverse of all explicit Weiner link arcs. See~\cite{belazzougui2017representing} for an extended description of the data structure and of the complexity of its operations. Here we just mention that the index supports operations $\mathtt{stringDepth}(\mathtt{id}(v))$ and $\mathtt{child}(\mathtt{id}(v),c)$ in constant time, and $\mathtt{parent}(\mathtt{id}(v))$, $\mathtt{suffixLink}(\mathtt{id}(v))$, $\mathtt{weinerLink}(\mathtt{id}(v),c)$ in $O(\log{\log{|T|}})$ time. 

In this paper we need to store the topology of $\REV{\SLT}$ and the topology of $\ST$ efficiently. It is well-known that the topology of an ordered tree of $n$ nodes can be represented using $2n+o(n)$ bits, as a sequence of $2n$ balanced parentheses \cite{munro2001succinct}. Let $\mathtt{id}(v)$ be the rank of a node $v$ in the preorder traversal of the tree. Given the balanced parentheses representation of the tree encoded in $2n+o(n)$ bits, it is also well-known that one can build a data structure that takes $2n+o(n)$ bits, and that supports several common operations in constant time \cite{navarro2016compact,NStalg14,SNsoda10}, among which: $\mathtt{parent}(\mathtt{id}(v))$, which returns $\mathtt{id}(u)$, where $u$ is the parent of $v$, or an error if $v$ is the root; $\mathtt{lca}(\mathtt{id}(v),\mathtt{id}(w))$, which returns $\mathtt{id}(u)$, where $u$ is the lowest common ancestor of nodes $v$ and $w$; $\mathtt{leftmostLeaf}(\mathtt{id}(v))$ and $\mathtt{rightmostLeaf}(\mathtt{id}(v))$, which return one plus the number of leaves that, in the preorder traversal of the tree, are visited before the first (respectively, the last) leaf that belongs to the subtree rooted at $v$; 
$\mathtt{depth}(\mathtt{id}(v))$, which returns the distance of $v$ from the root. This data structure can be built in $O(n)$ time and in $O(n)$ bits of working space. Moreover, given a node $v$ and a length $d$, a \emph{level-ancestor query} asks for the ancestor $u$ of $v$ such that the path from the root to $u$ contains exactly $d$ nodes. The level ancestor data structure described in \cite{bender2004level,berkman1994finding} takes $O(n)$ words of space and answers queries in constant time. Assuming that some nodes of the tree are marked, a \emph{lowest marked ancestor} data structure allows one to move in constant time from any node, to its lowest ancestor that is marked \cite{sokol2013engineering}.

We use the tree data structures described above to store the topology of $\ST$ and of $\REV{\SLT}$. Moreover, we mark in two bitvectors the nodes of $\REV{\SLT}$ and of $\ST$ that are maximal repeats (in preorder), and we index such bitvectors to support constant-time rank and select queries. Since $\REV{\SLT}$ is a subdivision of the subgraph of $\ST$ induced by maximal repeats, the $i$-th one in the two bitvectors correspond to the same maximal repeat. Thus, if node $v$ is a maximal repeat, and if we know its preorder position in $\ST$, we can compute the length of $\ell(v)$ by moving to the corresponding node $v'$ in $\REV{\SLT}$ and by computing the depth of $v'$ in the topology of $\REV{\SLT}$ (see \cite{belazzougui2017framework} for an extended explanation).

The rest of the paper focuses on representations of variable-order, bidirectional de Bruijn graphs that support the following primitives (for brevity we list here just operations in one direction). Let $k$ be the current order of the de Bruijn graph. Operation $v=\mathtt{node}(W)$, called \emph{membership}, 
returns the identifier of the node associated with $k$-mer $W$, or an error if $W$ does not occur in $T$. Operation $C=\mathtt{arcLabels}(v)$ returns the set of characters $C$ that label all arcs from node $v$ in the right direction, and operation $\mathtt{degree}(v)$ returns the number of such arcs. Query $e=\mathtt{arc}(v,c)$ returns the identifier of the arc that corresponds to string $\ell(v) \cdot c$, if any, where $v$ is a node in the current de Bruijn graph, $\ell(v)$ is the $k$-mer that corresponds to node $v$, and $c$ is a character; it returns an error if no such arc exists. Operation $w=\mathtt{followArc}(v,c)$ is similar, but returns the identifier of the node $w$ reached by the arc, if any.  Queries $\mathtt{freq}(v)$ and $\mathtt{freq}(e)$ return the number of occurrences of the $k$-mer associated with node $v$ and of the $(k+1)$-mer associated with arc $e$ (the number of occurrences of an arc might be zero). Representations that support such queries are called \emph{frequency-aware} or \emph{weighted} (see e.g. \cite{pandey2017debgr}). Operation $v'=\mathtt{increaseK}(v,c)$ for $c \in [0..\sigma]$ returns the node $v'$ associated with string $\ell(v) \cdot c$ in the de Bruijn graph of order $k+1$, if any, or an error otherwise. Operation $v'=\mathtt{decreaseK}(v)$ returns the node $v'$ associated with the prefix of length $k-1$ of $\ell(v)$ in the de Bruijn graph of order $k-1$.

In addition to increasing and decreasing the order by one unit, some variable-order representations allow the user to specify the desired amount of change \cite{belazzougui2016bidirectional,boucher2015variable}. 
In the rest of the paper we argue that it is more natural to change the order based on the frequency or on the extensions of $k$-mers, as proposed in \cite{diaz2018assembling}. Specifically, given a node $v$ of the current de Bruijn graph, let $\ell(v) \cdot W$, $W \in \Sigma^*$, be the longest string with the same frequency as $\ell(v)$ in $T$. Operation $(v',k')=\mathtt{increaseK}(v)$ returns the node $v'$ associated with $\ell(v) \cdot W$ in the de Bruijn graph of order $k+|W|$, and sets $k'$ to the new order $k+|W|$. Given a node $v$ of the current de Bruijn graph, let $W$ be the longest prefix of $\ell(v)$ that has a different frequency from $\ell(v)$ in $T$. Operation $(v',k')=\mathtt{decreaseK}(v)$ returns the node $v'$ associated with $W$ in the de Bruijn graph of order $|W|$, and sets $k'$ to $|W|$. Alternatively, one might want $W$ to be the longest prefix of $\ell(v)$ such that the left-extensions of $W$ are a superset of the left-extensions of $\ell(v)$. A de Bruijn graph that supports such operations without returning the value of the new order is called \emph{hidden-order}~\cite{diaz2018assembling}.

\subsection{String indexes} \label{sec:stringIndexes}

A \emph{bidirectional index} is a data structure that, given a constant-space descriptor $\mathtt{id}(W)$ of a substring $W$ of $T$, supports the following operations: $\mathtt{extendRight}(\mathtt{id}(W),a)=\mathtt{id}(Wa)$ if $f(Wa)>0$, or error otherwise; $\mathtt{enumerateRight}(\mathtt{id}(W))=\{\mathtt{id}(Wa) : a \in \Sigma, f(Wa)>0 \}$; $\mathtt{isRightMaximal}(\mathtt{id}(W))=\mathtt{true}$ iff $|\mathtt{enumerateRight}(\mathtt{id}(W))|>1$. Operations $\mathtt{extendLeft}$, $\mathtt{enumerateLeft}$ and $\mathtt{isLeftMaximal}$ are defined symmetrically. Here we consider bidirectional indexes based on the BWT: specifically, we denote with $\mathbb{I}(W,T)$ the function that maps a substring $W$ of $T$ to the interval of $W$ in $\BWT$, i.e. to the interval of all suffixes of $T\#$ that start with $W$, and we use $\mathtt{id}(W)=(\mathbb{I}(W,T),\mathbb{I}(\REV{W},\REV{T}),|W|)$ as a constant-space descriptor of a substring $W$. A number of bidirectional BWT indexes have been described in the literature: here we are interested just in the data structure described in \cite{BCKM13}, which supports all operations in linear time in the size of their output, takes $O(|T|\log{\sigma})$ bits of space, and can be built in randomized $O(|T|)$ time and $O(|T|\log{\sigma})$ bits of working space. See \cite{BCKM13} for more details.

Given a string $T \in [1..\sigma]^{n-1}\#$, we call \emph{run-length encoded BWT} ($\RLBWT_T$) any representation of $\BWT_T$ that takes $O(|\mathcal{R}_T|)$ words of space, and that supports the well known rank and select operations: see for example \cite{makinen2005succinct1,MakinenNSV10,SirenVMN08}. It is easy to implement a version of $\RLBWT_T$ that supports rank in $O(\log{\log{n}})$ time and select in $O(\log{\log{n}})$ time~\cite{belazzougui2015composite}. In this paper we use the representation of $\ST$ based on $\CDAWG$ described in \cite{belazzougui2017representing}, which takes just $O(e+\REV{e})$ words of space by augmenting $\CDAWG$ and $\REV{\CDAWG}$ with the RLBWT of $T$ and of $\REV{T}$. Such data structure represents a node $v$ of $\ST$ as a tuple $\mathtt{id}(v)=(v',|\ell(v)|,i,j)$, where $v'$ is the node in $\CDAWG$ that corresponds to the equivalence class of $v$, and $[i..j]$ is the interval of $\ell(v)$ in $\BWT$. For every node $v$ of $\CDAWG$, the index stores, among other things: $|\ell(v)|$ in a variable $v\mathtt{.length}$; the number $v\mathtt{.size}$ of right-maximal strings that belong to its equivalence class; and the interval $[v\mathtt{.first}..v\mathtt{.last}]$ of $\ell(v)$ in $\BWT_T$. For every arc $\gamma=(v,w)$ of $\CDAWG$, the index stores the first character of $\ell(\gamma)$ in a variable $\gamma\mathtt{.char}$, and the number of characters of the right extension implied by $\gamma$ in a variable $\gamma\mathtt{.right}$. Finally, we add to the CDAWG all arcs $(v,w,c)$ such that $w$ is the equivalence class of the destination of a Weiner link from $v$ labeled by character $c$ in $\ST_T$, and the reverse of all explicit Weiner link arcs. See \cite{belazzougui2017representing} for a full description of the data structure and of the complexity of its operations. Here we just mention that the index supports operations $\mathtt{stringDepth}(\mathtt{id}(v))$ and $\mathtt{child}(\mathtt{id}(v),c)$ in constant time, and $\mathtt{parent}(\mathtt{id}(v))$, $\mathtt{suffixLink}(\mathtt{id}(v))$, $\mathtt{weinerLink}(\mathtt{id}(v))$ in $O(\log{\log{|T|}})$ time. It also allows reading the character at position $i$ of $T$ in $O(\log{|T|})$ time.

Finally, in this paper we need to store the topology of $\REV{\SLT}$ and the topology of $\ST$ efficiently. It is well known that the topology of an ordered tree of $n$ nodes can be represented using $2n+o(n)$ bits, as a sequence of $2n$ balanced parentheses built by opening a parenthesis, by recurring on every child of the current node in order, and by closing a parenthesis~\cite{munro2001succinct}. Let $\mathtt{id}(v)$ be the rank of a node $v$ in the preorder traversal of the tree. Given the balanced parentheses representation of the tree encoded in $2n+o(n)$ bits, it is also well known that one can build a data structure that takes $2n+o(n)$ bits, and that supports several common operations in constant time \cite{navarro2016compact,SNsoda10,NStalg14}, among which: $\mathtt{parent}(\mathtt{id}(v))$, which returns $\mathtt{id}(u)$, where $u$ is the parent of $v$, or an error if $v$ is the root; $\mathtt{lca}(\mathtt{id}(v),\mathtt{id}(w))$, which returns $\mathtt{id}(u)$, where $u$ is the lowest common ancestor of nodes $v$ and $w$; $\mathtt{leftmostLeaf}(\mathtt{id}(v))$ and $\mathtt{rightmostLeaf}(\mathtt{id}(v))$, which return one plus the number of leaves that, in the preorder traversal of the tree, are visited before the first (respectively, the last) leaf that belongs to the subtree rooted at $v$; $\mathtt{selectLeaf}(i)$, which returns $\mathtt{id}(v)$, where $v$ is the $i$-th leaf in preorder; $\mathtt{depth}(\mathtt{id}(v))$, which returns the distance of $v$ from the root. This data structure can be built in $O(n)$ time and in $O(n)$ bits of working space. Moreover, given a node $v$ and a length $d$, a \emph{level-ancestor query} asks for the ancestor $u$ of $v$ such that the path from the root to $u$ contains exactly $d$ nodes. The level ancestor data structure described in \cite{bender2004level,berkman1994finding} takes $O(n)$ words of space and it answers queries in constant time. Assuming that some nodes of the tree are marked, a \emph{lowest marked ancestor} data structure \cite{sokol2013engineering} allows one to move in constant time from any node, to its lowest ancestor that is marked.

We use the tree data structures described above to store the topology of $\ST$ and of $\REV{\SLT}$. Moreover, we mark in a bitvector the nodes of $\REV{\SLT}$ and of $\ST$ that are maximal repeats (in preorder), and we index such bitvectors to support constant-time rank and select queries. Since $\REV{\SLT}$ is a subdivision of the subgraph of $\ST$ induced by maximal repeats, the $i$-th one in the two bitvectors correspond to the same maximal repeat. Thus, if node $v$ is a maximal repeat and if we know its position in preorder in $\ST$, it is easy to see that we can compute the length of $\ell(v)$ by going to the node $v'$ in $\REV{\SLT}$ and by computing the depth of $v'$ in the topology of $\REV{\SLT}$: see \cite{belazzougui2017framework} for a more thorough explanation.

\section{Contracting in constant time} \label{sec:contract}

As mentioned, existing bidirectional BWT indexes support left-contraction just from right-maximal substrings (and symmetrically, they support right-contraction just from left-maximal substrings). Specifically, if the substring $aW$ is right-maximal and labels a node $v$ of $\ST$, then $\mathbb{I}(W,T)$ is the interval of node $\mathtt{suffixLink}(v)$ in $\ST$, and since we are removing one character from the right of $\REV{aW}$, the locus of $\REV{W}$ in $\REV{\ST}$ is either the same as the locus $w$ of $\REV{aW}$, or it is $\mathtt{parent}(w)$, whichever has the same frequency as $\mathbb{I}(W,T)$ \cite{BCKM13,munro2017space}.

To support left-contraction from a substring that is not right-maximal, it is enough to have access to the topology of $\REV{\SLT}$:

\begin{theorem} \label{thm:contract}
Let $T$ be a string on alphabet $\Sigma$. There is a data structure that supports operations $\ExtendRight$, $\ExtendLeft$, $\ContractRight$ and $\ContractLeft$ in constant time and in $O(n\log\sigma)$ bits of space. Such a data structure can be built in randomized $O(n)$ time and $O(n\log\sigma)$ bits of working space.
\end{theorem}
\begin{proof}
We use the data structures described in \cite{BCKM13}, augmented with the topology of $\SLT$ and with a bitvector to commute between the topology of $\ST$ and the topology of $\SLT$ (see~\cite{belazzougui2017framework} for details on commuting). Such data structures take $O(n\log\sigma)$ bits of space, and they can be built in randomized $O(n)$ time using the algorithms in \cite{belazzougui2014linear,belazzougui2016linear}. They support operations $\mathtt{extendRight}(\mathtt{id}(W),a)=\mathtt{id}(Wa)$ and $\mathtt{extendLeft}(\mathtt{id}(W),a)=\mathtt{id}(aW)$, where $\mathtt{id}(W)=(\mathbb{I}(W,T),\mathbb{I}(\REV{W},\REV{T}))$. We additionally assume the knowledge of $|W|$, i.e. $\mathtt{id}(W)=(\mathbb{I}(W,T),\mathbb{I}(\REV{W},\REV{T}),|W|)$. We only show how to support 
$\mathtt{contractLeft}(\mathtt{id}(aW))=\mathtt{id}(W)$, since supporting $\mathtt{contractRight}(\mathtt{id}(Wa))=\mathtt{id}(W)$ is symmetric. Since \cite{BCKM13} already supports $\mathtt{contractLeft}(\mathtt{id}(aW))$ for right-maximal substrings, we assume for now that $aW$ is not right-maximal. Note that we can decide whether $aW$ is right-maximal or not by using $\mathbb{I}(\REV{aW},\REV{T})$, and, if $W$ is right-maximal, we can just use the contraction algorithm described above. Let $v$ be the locus of $aW$ in $\ST$: this can be computed from $\mathbb{I}(aW,T)$ using $\mathtt{lca}$ queries on $\ST$. Since $aW$ is not right maximal, $aW \neq \ell(v)$ and $aW$ ends in the middle of edge $(u,v)$ of $\ST$. We take in constant time the suffix link $(u,u')$ from $u$ and the suffix link $(v,v')$ from $v$, and we decide whether $(u',v')$ is an edge or a path of $\ST$ by comparing $u'$ to $\mathtt{parent}(v')$, which can be computed in constant time. If $(u',v')$ is an edge of $\ST$ (Figure~\ref{fig:contract}c), then $v'$ is the locus of $W$ and we compute $\mathbb{I}(\ell(v'),T)$ in constant time. Otherwise (Figure~\ref{fig:contract}d), we compute in constant time $z=\mathtt{parent}(v')$: this node is a maximal repeat by Property~\ref{property:maxRepeats}, since it is an internal node of $\ST$ with an implicit Weiner link whose destination falls inside $(u,v)$. We use the data structures in Section~\ref{sec:stringIndexes} to measure the length of $\ell(z)$ in constant time. If $|W|>|\ell(z)|$, the locus of $W$ is again $v'$. Otherwise, since $z$ is a maximal repeat, we move in constant time to the node $z'$ of $\REV{\SLT}$ that corresponds to $\ell(z)$, we issue a constant-time level ancestor query from $z'$ on $\REV{\SLT}$ with length $|W|$, and, from the destination $x'$ of such a level ancestor query, we move in constant time to the first branching descendant $y'$ of $x'$, by using $\mathtt{leftmostLeaf}$, $\mathtt{rightmostLeaf}$, and $\mathtt{lca}$ queries on $\REV{\SLT}$. Finally, we move in constant time to the node $y$ of $\ST$ that corresponds to $y'$, and we compute $\mathbb{I}(\ell(y),T)$ in constant time. We compute $\mathbb{I}(\REV{W},\REV{\ST})$ as described at the beginning of Section \ref{sec:contract}.
\end{proof}

Note that the algorithm in Theorem~\ref{thm:contract} works even when $aW$ is right-maximal; moreover, if the information on whether $aW$ is right maximal or not is given in input, the algorithm can decide whether $W$ is right maximal or not. In a practical implementation, once we have taken the suffix link $(v,v')$ from $v$, we could check whether $v'$ is a maximal repeat, and in the positive case we could immediately commute to $\REV{\SLT}$ and issue level ancestor queries. If $v'$ is not a maximal repeat, we could move in constant time to the lowest ancestor $v''$ of $v'$ that is a maximal repeat, using a lowest marked ancestor data structure on $\ST$, we could measure $|\ell(v'')|$, and if $|\ell(v'')| \geq |W|$, we could again issue level ancestor queries in $\REV{\SLT}$ (otherwise, the locus of $W$ is again $v'$).

A bidirectional index on $T$ that supports extension and contraction in constant time, can be used to implement in linear time several applications that slide a window $S[i..j]$ of fixed length over a query string $S$, and that compute the frequency of every $S[i..j]$ in $T$, \emph{without the size of the window being known during construction}\footnote{If the size $k$ of the window is fixed and known during construction, most such applications do not need the contract operation, and can be made to work using just one BWT and a bitvector of length $|T|$ that marks the boundaries of $k$-mer intervals in the BWT.}. For example, measuring the frequency of windows of fixed length for read correction \cite{philippe2013crac}, computing the inner product between the $k$-mer composition vectors of $S$ and $T$ (a step in $k$-mer kernels), estimating the probability of $S$ according to a fixed-order Markov model trained on $T$, and checking whether $S$ is a path in the de~Bruijn graph of $T$. Our index enables also applications in which \emph{the sliding window needs to be extended or contracted during the scan}, like variable-order and interpolated Markov models (see \cite{cunial2018framework} for an overview). A fully-functional bidirectional index is not needed for computing the matching statistics array between $S$ and $T$, in linear time and in $O(|T|\log\sigma)$ bits of space, since one can use the algorithms in \cite{belazzougui2014indexed} on top of the data structures in \cite{belazzougui2014linear}. However, achieving such bounds with our bidirectional index becomes trivial.

In practical applications of matching statistics, one typically needs to maintain the intervals in both $\BWT$ and $\REV{\BWT}$ just after every successful right extension, and, when the current match $S[i..j]$ cannot be extended with $S[j+1]$ in $T$ any longer, one might need both BWT intervals just for the proper suffixes $S[k..j]$ such that $\Sigma^{r}_{T}(S[i..j]) \subset \Sigma^{r}_{T}(S[k..j])$, i.e. just for the suffixes of $S[i..j]$ from which a right-extension with $S[j+1]$ is attempted again. Every such suffix is a maximal repeat ancestor of $\REV{S[i..j]}$ in $\REV{\ST}$ \cite{belazzougui2018fast}, thus, once we reach the locus of such a suffix in $\REV{\ST}$ with $\mathtt{parent}$ operations, we can compute its interval in $\REV{\BWT}$, we can measure its string length $p$, and we can compute its interval in $\BWT$ by issuing $\MS[i]-p$ contract operations from the locus of $S[i..j]$ in $\ST$, but without updating the interval in $\REV{\BWT}$ after each contraction. Even more aggressively, we can just issue $\MS[i]-p$ suffix links from the locus of $S[i..j]$ in $\ST$. Note that such a locus might correspond to the right-maximal string $S[i..j] \cdot V$ for some nonempty $V$, thus taking $\MS[i]-p$ suffix links might lead to a node of $\ST$ that corresponds to the right-maximal string $S[k..j] \cdot V$: thus, we need to move in constant time from such a node, to its lowest ancestor in $\ST$ that is a maximal repeat; from there, we can then issue a level ancestor query with value $p$. Such a lazy synchronization might be faster than issuing $\MS[i]-p$ full contract operations in practice. 
%
%
%

Our index can be seen as a representation of a de~Bruijn graph that supports bidirectional navigation, that allows access to the frequency of every $k$-mer and $(k+1)$-mer, and that has no upper bound on the order: we call \emph{infinite-order} such a de~Bruijn graph. Note that, for a given order $k$, we can support both the variant in which arcs must occur in $T$ (calling $\mathtt{extendRight}$ and then $\mathtt{contractLeft}$ to implement $\mathtt{arc}$ and $\mathtt{followArc}$), and the variant in which arcs do not have to occur in $T$ (calling $\mathtt{contractLeft}$ and then $\mathtt{extendRight}$). Membership queries reduce to backward searches, and we can move from a higher to a lower order using the same algorithm as in matching statistics. Indeed, one typically wants to switch to a suffix of the current $k$-mer whenever there is only one arc in the graph of the current order, and this arc is labelled with the terminator character \cite{diaz2018assembling}; or, more generally, whenever one needs to increase the number of outgoing arcs from the current $k$-mer (for example because the existing ones have already been explored \cite{morisse2018hybrid}), or to increase the frequency of the current right-maximal $k$-mer. In all such cases, one wants to switch to the largest order with the desired property, and the corresponding suffix is always a maximal repeat (for example, the longest suffix, of the current right-maximal $k$-mer, that has strictly greater frequency, is a maximal repeat). Symmetrically, when increasing the order, one may want to switch e.g. from the current $k$-mer $W$ that is left-maximal but not right-maximal, to the maximal repeat $WV$ with shortest $V$. Clearly $\mathbb{I}(WV,T)=\mathbb{I}(W,T)$, we know $|V|$ since we can access $|WV|$, and we can compute $\mathbb{I}(\REV{WV},\REV{T})$ by taking $|V|$ Weiner links from $\mathbb{I}(\REV{W},\REV{T})$. All such Weiner links are explicit, and in practice we can just update the first position of the interval at every step. 
%

In the next section, we describe a representation of an infinite-order de~Bruijn graph in which the time to decrease or increase the order does not depend on the difference between the source and the destination order.

\section{Implementing de Bruijn graphs with CDAWGs} \label{sec:dbgCDAWG}

An \emph{affix link} $\mathbb{A}(w)$ is a map from a node $w$ of $\ST$, to the locus of $\REV{\ell(w)}$ in $\REV{\ST}$ (we use $\REV{\mathbb{A}}(w)$ to denote the symmetrical map from a node $w$ of $\REV{\ST}$, to the locus of $\REV{\ell(w)}$ in $\ST$) \cite{stoye2000affix,strothmann2007affix}. We use $\mathbb{A}(W)$ as a shorthand for $\mathbb{A}(w)$ where $w$ is the locus of $W$. In asynchronous bidirectional indexes, affix links are used to switch direction when the user desires \cite{strothmann2007affix}. In this section we are more interested in their ability to extend a non-maximal repeat in a bidirectional index: for example, if $W$ is right-maximal but not left-maximal, and if it has loci $(v,w)$ in $\ST$ and $\REV{\ST}$, respectively, then its shortest left-maximal extension $VW$ with $|V| \geq 0$, i.e. the shortest maximal repeat that contains $W$ as a (not necessarily proper) suffix, has loci $(\REV{\mathbb{A}}(w),w)$; and if $W$ is neither left- nor right-maximal, then the shortest maximal repeat $UWV$ with the same frequency as $W$ has loci $(\REV{\mathbb{A}}(\mathbb{A}(v)),\mathbb{A}(v)) = (\REV{\mathbb{A}}(w),\mathbb{A}(\REV{\mathbb{A}}(w)))$ \cite{strothmann2007affix}. 
%
%
%
Thus, in what follows we ignore affix links from leaves.

Rather than storing $\mathbb{A}(w)$ for every internal node $w$ of $\ST$, it has been proposed to sample $\mathbb{A}(w)$ every $p$ suffix links \cite{canovas2017full}: indeed, $\mathbb{A}(w)$ is either $v=\mathbb{A}(\mathtt{suffixLink}(w))$, if $|\ell(v)| \geq |\ell(w)|$, or it is the child of $v$ obtained by following the first character of $\ell(w)$ \cite{strothmann2007affix}. 
This allows one to compute $\mathbb{A}(w)$ in $O(p)$ time, paying $O((|T|/p)\log{n})$ bits of space. We briefly observe that, compared to existing sampling schemes for bidirectional indexes, we can further reduce space to $O((|T|/p)\log{m})$ bits, where $m$ is the number of maximal repeats of $T$, since, by Property~\ref{property:affixLinks}, $\mathbb{A}(v)$ is a maximal repeat of $T$ for every internal node $v$ of $\ST_T$. In practice following Weiner links is faster than following suffix links: thus, one could sample the value of $\mathbb{A}(w)$ for every maximal repeat, and then sample every $p$ characters inside an edge of $\REV{\ST}$ that connects two maximal repeats, i.e. every $p$ explicit Weiner links. If $\mathbb{A}(w)$ is not sampled, then $\ell(w)$ is not left-maximal, so we take the only possible Weiner link from it and we repeat the search from there, returning the value of the first sampled node we find. This sampling scheme takes $O((m+(|T|-m)/p)\log{m})$ bits of space. One could even waive sampling the nodes of $\ST$ that are not maximal repeats, but to retrieve their value one would have to pay a number of Weiner links that is at most equal to the length of the longest edge of $\REV{\ST}$ connecting two maximal repeats. Clearly, sampling just maximal repeats works also for the scheme based on suffix links. 
%
%

In this section we store $\mathbb{A}(w)$ and $\REV{\mathbb{A}}(w)$ explicitly, but just for maximal repeats, together with $\CDAWG_T$ and $\REV{\CDAWG}_T$, to implement an infinite-order de~Bruijn graph in which the time to increase or decrease the order does not depend on the difference between the source and the destination order:

\begin{theorem} \label{thm:cdawg}
Given a string $T$, there are a fully-functional bidirectional index, and an infinite-order representation of the de Bruijn graph of $T$, that take space proportional to the number of left and right extensions of the maximal repeats of $T$, and that support all queries in $O(\log{\log{|T|}})$ time.
\end{theorem}
\begin{proof}
We represent $\ST$ and $\REV{\ST}$ using CDAWGs, as described in \cite{belazzougui2017representing} and summarized in Section \ref{sec:stringIndexes} of this paper. In addition to $\RLBWT$, $\REV{\RLBWT}$, $\CDAWG$ and $\REV{\CDAWG}$, to support Theorem~\ref{thm:contract} we store also a weighted level ancestor data structure on the maximal repeat subgraph of $\ST$ and $\REV{\ST}$, which takes $O(m)$ space and answers queries in $O(\log{\log{|T|}})$ time \cite{amir2007dynamic,farach1996perfect}, and we store $\mathbb{A}$ and $\REV{\mathbb{A}}$ to support changes in the order of the de Bruijn graph ($m$ is the number of maximal repeats of $T$). 
%
We represent an arbitrary substring $W$ of $T$ as a triple $(\mathtt{id}(v),\mathtt{id}(w),|W|)$, where $v$ is the locus of $W$ in $\ST$, $w$ is the locus of $\REV{W}$ in $\REV{\ST}$, and $\mathtt{id}$ is the identifier of a node in the CDAWG-based representation of a suffix tree, i.e. $\mathtt{id}(v)=(v',|\ell(v)|,i,j)$ where $v'$ is a node of a CDAWG and $[i..j]$ is a BWT interval.

To implement $\mathtt{extendRight}(W,c)$, where $Wc$ is assumed to occur in $T$, we first check whether $W$ is right-maximal, by comparing $|W|$ to $|\ell(v)|$: if $W$ is not right-maximal, then the representation of $Wc$ is $(\mathtt{id}(v),\mathtt{weinerLink}(\mathtt{id}(w),c),|W|+1)$. Otherwise, the representation is $(\mathtt{child}(\mathtt{id}(v),c),\mathtt{weinerLink}(\mathtt{id}(w),c),|W|+1)$. If we assume that procedure $\mathtt{extendRight}(W,c)$ can be called with an invalid $c$, we first have to check whether $Wc$ occurs in $T$ using the interval of $\REV{W}$ in $\REV{\BWT}$. To implement $\mathtt{contractLeft}(aW)$, we first check whether $aW$ is right-maximal, by comparing $|aW|$ to $|\ell(v)|$: if so, the representation of $W$ is $(\mathtt{suffixLink}(\mathtt{id}(v)),\mathtt{id}(w'),|W|)$, where $w'$ is either the parent of $w$ or $w$ itself, depending on which one of them has the same frequency as the locus of $W$ in $\ST$. If $aW$ is not right-maximal, we run the algorithm in Theorem~\ref{thm:contract} using the $\mathtt{suffixLink}$ and $\mathtt{parent}$ operations provided by the CDAWG-based representation of $\ST$, and issuing weighted level ancestor queries on the maximal repeat subgraph of $\mathsf{ST}$ rather than level ancestor queries on the topology of $\REV{\mathsf{SLT}}$.

To implement $\mathtt{decreaseK}$ and $\mathtt{increaseK}$ in the de Bruijn graph, we proceed as follows. If the current $k$-mer $W$ is right-maximal, the representation of the longest suffix of $W$ that is a maximal repeat is clearly $(\mathtt{id}(z),\mathtt{id}(\mathbb{A}(z)),|\ell(z)|)$, where $z$ is the maximal repeat reached by taking a suffix link arc from the node of the CDAWG pointed by $\mathtt{id}(v)$. One could further move to a suitable ancestor of such a maximal repeat, by marking the topology of the maximal repeat subgraph of $\ST$. If the current $W$ is left-maximal but not right-maximal, the representation of the shortest maximal repeat of the form $WV$ for some nonempty $V$ is $(\mathtt{id}(z),\mathtt{id}(\mathbb{A}(z)),|\ell(z)|)$, where $z$ is the node of the CDAWG pointed by $\mathtt{id}(v)$. The same holds if $W$ is neither left- nor right-maximal, and if we want to move to the shortest $k$-mer that contains $W$ and is both left- and right-maximal. Implementing the other operations of a bidirectional de Bruijn graph is straightforward and is left to the reader. We use data structures from \cite{belazzougui2017fast} to answer the membership query $\mathtt{node}(W)$ in $O(|W|)$ time.
\end{proof}
%
%

Our construction based on two CDAWGs is reminiscent of the \emph{symmetric compact DAWG} described in \cite{blumer1987complete}, which was used however just for bidirectional extension. 
Theorem~\ref{thm:cdawg} could be simplified in several ways for a practical implementation. For example, as noted already in \cite{blumer1987complete}, since $\CDAWG$ and $\REV{\CDAWG}$ share the same set of nodes, every such node could be stored only once, in which case $\mathbb{A}$ and $\REV{\mathbb{A}}$ would not need to be represented explicitly. If the descriptor of a substring $W$ is $(\mathtt{id}(v),\mathtt{id}(w),|W|)$ with $\mathtt{id}(v)=(v',|\ell(v)|,i,j)$ and $\mathtt{id}(w)=(w',|\ell(w)|,i',j')$, then $v'$ and $w'$ would become pointers to the same node, $|\ell(w)|$ could be derived from $|\ell(v')|-|\ell(v)|+|W|$, and rather than storing $i,j$ and $i',j'$, we could just store $i,i',f(W)$. 
Our representation collapses to the sink of a CDAWG all $k$-mers that occur just once in the dataset, which are likely induced by sequencing errors and are thus not useful for most applications: in this case, we don't even need to store left and right extensions of maximal repeats directed to the sink. If the target application never uses orders smaller than a threshold $\tau$, we could remove from the index all maximal repeats of length smaller than $\tau$ and prune the top part of the corresponding tree data structures, as described in \cite{diaz2018assembling}. We could proceed in a similar way when the user specifies a lower bound on the frequency of $k$-mers (called \emph{solid}, see e.g. \cite{li2016megahit,morisse2018hybrid}). 
%
%

\section{Discussion and extensions}

Our CDAWG-based representation of the de Bruijn graph might be practical: a full experimental study and a careful implementation of each primitive would be an interesting research direction. Given a node $v$ in the de Bruijn graph, it would also be interesting to know if we can traverse an entire maximal non-branching path, i.e. a path in which no $k$-mer except for $v$ and the destination has more than one arc to the left and to the right, without taking time proportional to the length of such a path: this would provide a fast implementation of the \emph{compacted} de Bruijn graph (see e.g. \cite{chikhi2016compacting,minkin2016twopaco} and references therein). 
It is natural to wonder whether one can support the operations of an infinite-order de~Bruijn graph in less space than our indexes. Another open question is whether the CDAWG can be used as a substrate for implementing the \emph{string graph} as well, and whether we can design a single compact index, as wished by \cite{diaz2019simulating}, that supports both the primitives of a string graph and of an infinite-order de Bruijn graph efficiently, allowing the user to take advantage of both approaches in genome assembly.

\section{Acknowledgements}

We thank Martin Bundgaard for motivating the contract operation, Rodrigo Canovas for discussions about bidirectional indexes, Gene Myers for discussions about PacBio CCS reads, and German Tischler for help with $k$-mer counting.

\bibliographystyle{plain}
\bibliography{bbwt}

\begin{thebibliography}{10}

\bibitem{amir2007dynamic}
Amihood Amir, Gad~M Landau, Moshe Lewenstein, and Dina Sokol.
\newblock Dynamic text and static pattern matching.
\newblock {\em ACM Transactions on Algorithms (TALG)}, 3(2):19, 2007.

\bibitem{apostolico2000optimal}
Alberto Apostolico and Gill Bejerano.
\newblock Optimal amnesic probabilistic automata or how to learn and classify
  proteins in linear time and space.
\newblock {\em Journal of Computational Biology}, 7(3-4):381--393, 2000.

\bibitem{baier2015graphical}
Uwe Baier, Timo Beller, and Enno Ohlebusch.
\newblock Graphical pan-genome analysis with compressed suffix trees and the
  {Burrows--Wheeler} transform.
\newblock {\em Bioinformatics}, 32(4):497--504, 2015.

\bibitem{belazzougui2014linear}
Djamal Belazzougui.
\newblock Linear time construction of compressed text indices in compact space.
\newblock In {\em Proceedings of the forty-sixth Annual ACM Symposium on Theory
  of Computing}, pages 148--193. ACM, 2014.

\bibitem{belazzougui2014indexed}
Djamal Belazzougui and Fabio Cunial.
\newblock Indexed matching statistics and shortest unique substrings.
\newblock In {\em International Symposium on String Processing and Information
  Retrieval}, pages 179--190. Springer, 2014.

\bibitem{belazzougui2017fast}
Djamal Belazzougui and Fabio Cunial.
\newblock Fast label extraction in the {CDAWG}.
\newblock In {\em International Symposium on String Processing and Information
  Retrieval}, pages 161--175. Springer, 2017.

\bibitem{belazzougui2017framework}
Djamal Belazzougui and Fabio Cunial.
\newblock A framework for space-efficient string kernels.
\newblock {\em Algorithmica}, 79(3):857--883, 2017.

\bibitem{belazzougui2017representing}
Djamal Belazzougui and Fabio Cunial.
\newblock Representing the suffix tree with the {CDAWG}.
\newblock In {\em LIPIcs-Leibniz International Proceedings in Informatics},
  volume~78. Schloss Dagstuhl-Leibniz-Zentrum fuer Informatik, 2017.

\bibitem{belazzougui2018fast}
Djamal Belazzougui, Fabio Cunial, and Olgert Denas.
\newblock Fast matching statistics in small space.
\newblock In {\em LIPIcs-Leibniz International Proceedings in Informatics},
  volume 103. Schloss Dagstuhl-Leibniz-Zentrum fuer Informatik, 2018.

\bibitem{belazzougui2015composite}
Djamal Belazzougui, Fabio Cunial, Travis Gagie, Nicola Prezza, and Mathieu
  Raffinot.
\newblock Composite repetition-aware data structures.
\newblock In {\em Annual Symposium on Combinatorial Pattern Matching}, pages
  26--39. Springer, 2015.

\bibitem{BCKM13}
Djamal Belazzougui, Fabio Cunial, Juha K{\"a}rkk{\"a}inen, and Veli
  M{\"a}kinen.
\newblock Versatile succinct representations of the bidirectional
  {Burrows-Wheeler} transform.
\newblock In {\em 21st Annual European Symposium on Algorithms (ESA 2013)},
  volume 8125 of {\em Lecture Notes in Computer Science}, pages 133--144,
  France, 2013. Springer.

\bibitem{belazzougui2016linear}
Djamal Belazzougui, Fabio Cunial, Juha K{\"a}rkk{\"a}inen, and Veli
  M{\"a}kinen.
\newblock Linear-time string indexing and analysis in small space.
\newblock {\em arXiv preprint arXiv:1609.06378}, 2016.

\bibitem{belazzougui2016bidirectional}
Djamal Belazzougui, Travis Gagie, Veli M{\"a}kinen, Marco Previtali, and
  Simon~J Puglisi.
\newblock Bidirectional variable-order {de Bruijn} graphs.
\newblock In {\em Latin American Symposium on Theoretical Informatics}, pages
  164--178. Springer, 2016.

\bibitem{bender2004level}
Michael~A Bender and Mart{\i}n Farach-Colton.
\newblock The level ancestor problem simplified.
\newblock {\em Theoretical Computer Science}, 321(1):5--12, 2004.

\bibitem{berkman1994finding}
Omer Berkman and Uzi Vishkin.
\newblock Finding level-ancestors in trees.
\newblock {\em Journal of Computer and System Sciences}, 48(2):214--230, 1994.

\bibitem{blumer1987complete}
Anselm Blumer, Janet Blumer, David Haussler, Ross McConnell, and Andrzej
  Ehrenfeucht.
\newblock Complete inverted files for efficient text retrieval and analysis.
\newblock {\em Journal of the {ACM}}, 34(3):578--595, 1987.

\bibitem{boucher2015variable}
Christina Boucher, Alex Bowe, Travis Gagie, Simon~J Puglisi, and Kunihiko
  Sadakane.
\newblock Variable-order de {Bruijn} graphs.
\newblock In {\em 2015 Data Compression Conference}, pages 383--392. IEEE,
  2015.

\bibitem{canovas2017full}
Rodrigo C{\'a}novas and Eric Rivals.
\newblock Full compressed affix tree representations.
\newblock In {\em Data Compression Conference (DCC), 2017}, pages 102--111.
  IEEE, 2017.

\bibitem{chikhi2016compacting}
Rayan Chikhi, Antoine Limasset, and Paul Medvedev.
\newblock Compacting de bruijn graphs from sequencing data quickly and in low
  memory.
\newblock {\em Bioinformatics}, 32(12):i201--i208, 2016.

\bibitem{CrochemoreV97}
Maxime Crochemore and Renaud V\'erin.
\newblock Direct construction of compact directed acyclic word graphs.
\newblock In Alberto Apostolico and Jotun Hein, editors, {\em CPM}, volume 1264
  of {\em Lecture Notes in Computer Science}, pages 116--129. Springer, 1997.

\bibitem{cunial2018framework}
Fabio Cunial, Jarno Alanko, and Djamal Belazzougui.
\newblock A framework for space-efficient variable-order {Markov} models.
\newblock {\em bioRxiv preprint}, page 443101, 2018.

\bibitem{diaz2018assembling}
Diego D{\'\i}az-Dom{\'\i}nguez, Djamal Belazzougui, Travis Gagie, Veli
  M{\"a}kinen, Gonzalo Navarro, and Simon~J Puglisi.
\newblock Assembling omnitigs using hidden-order {de Bruijn} graphs.
\newblock {\em arXiv preprint arXiv:1805.05228}, 2018.

\bibitem{diaz2019simulating}
Diego D{\'\i}az-Dom{\'\i}nguez, Travis Gagie, and Gonzalo Navarro.
\newblock Simulating the {DNA} string graph in succinct space.
\newblock {\em arXiv preprint arXiv:1901.10453}, 2019.

\bibitem{farach1996perfect}
Martin Farach and S~Muthukrishnan.
\newblock Perfect hashing for strings: formalization and algorithms.
\newblock In {\em Annual Symposium on Combinatorial Pattern Matching}, pages
  130--140. Springer, 1996.

\bibitem{gog2014multi}
Simon Gog, Kalle Karhu, Juha K{\"a}rkk{\"a}inen, Veli M{\"a}kinen, and Niko
  V{\"a}lim{\"a}ki.
\newblock Multi-pattern matching with bidirectional indexes.
\newblock {\em Journal of Discrete Algorithms}, 24:26--39, 2014.

\bibitem{koch2014repark}
Philipp Koch, Matthias Platzer, and Bryan~R Downie.
\newblock {RepARK} — de novo creation of repeat libraries from whole-genome
  {NGS} reads.
\newblock {\em Nucleic Acids Research}, 42(9):e80--e80, 2014.

\bibitem{kokot2017kmc}
Marek Kokot, Maciej D{\l}ugosz, and Sebastian Deorowicz.
\newblock {KMC} 3: counting and manipulating $k$-mer statistics.
\newblock {\em Bioinformatics}, 33(17):2759--2761, 2017.

\bibitem{lam2009high}
Tak~Wah Lam, Ruiqiang Li, Alan Tam, Simon Wong, Edward Wu, and Siu-Ming Yiu.
\newblock High throughput short read alignment via bi-directional {BWT}.
\newblock In {\em Bioinformatics and Biomedicine, 2009. BIBM'09. IEEE
  International Conference on}, pages 31--36. IEEE, 2009.

\bibitem{li2016megahit}
Dinghua Li, Ruibang Luo, Chi-Man Liu, Chi-Ming Leung, Hing-Fung Ting, Kunihiko
  Sadakane, Hiroshi Yamashita, and Tak-Wah Lam.
\newblock {MEGAHIT} v1.0: a fast and scalable metagenome assembler driven by
  advanced methodologies and community practices.
\newblock {\em Methods}, 102:3--11, 2016.

\bibitem{maass2000linear}
Moritz~G Maa{\ss}.
\newblock Linear bidirectional on-line construction of affix trees.
\newblock In {\em Annual Symposium on Combinatorial Pattern Matching}, pages
  320--334. Springer, 2000.

\bibitem{makinen2005succinct1}
Veli M{\"a}kinen and Gonzalo Navarro.
\newblock Succinct suffix arrays based on run-length encoding.
\newblock In {\em Combinatorial Pattern Matching}, pages 45--56. Springer,
  2005.

\bibitem{MakinenNSV10}
Veli M{\"{a}}kinen, Gonzalo Navarro, Jouni Sir{\'{e}}n, and Niko
  V{\"{a}}lim{\"{a}}ki.
\newblock Storage and retrieval of highly repetitive sequence collections.
\newblock {\em Journal of Computational Biology}, 17(3):281--308, 2010.

\bibitem{sokol2013engineering}
Shoshana Marcus and Dina Sokol.
\newblock Engineering small space dictionary matching.
\newblock {\em arXiv preprint arXiv:1301.6428}, 2013.

\bibitem{mauri2003pattern}
Giancarlo Mauri and Giulio Pavesi.
\newblock Pattern discovery in {RNA} secondary structure using affix trees.
\newblock In {\em Annual Symposium on Combinatorial Pattern Matching}, pages
  278--294. Springer, 2003.

\bibitem{minkin2019scalable}
Ilia Minkin and Paul Medvedev.
\newblock Scalable multiple whole-genome alignment and locally collinear block
  construction with {SibeliaZ}.
\newblock {\em bioRxiv preprint}, page 548123, 2019.

\bibitem{minkin2016twopaco}
Ilia Minkin, Son Pham, and Paul Medvedev.
\newblock {TwoPaCo}: An efficient algorithm to build the compacted de {Bruijn}
  graph from many complete genomes.
\newblock {\em Bioinformatics}, 33(24):4024--4032, 2016.

\bibitem{morisse2018hybrid}
Pierre Morisse, Thierry Lecroq, and Arnaud Lefebvre.
\newblock Hybrid correction of highly noisy long reads using a variable-order
  {de Bruijn} graph.
\newblock {\em Bioinformatics}, 34(24):4213--4222, 2018.

\bibitem{munro2017space}
J~Ian Munro, Gonzalo Navarro, and Yakov Nekrich.
\newblock Space-efficient construction of compressed indexes in deterministic
  linear time.
\newblock In {\em Proceedings of the Twenty-Eighth Annual ACM-SIAM Symposium on
  Discrete Algorithms}, pages 408--424. SIAM, 2017.

\bibitem{munro2001succinct}
J~Ian Munro and Venkatesh Raman.
\newblock Succinct representation of balanced parentheses and static trees.
\newblock {\em SIAM Journal on Computing}, 31(3):762--776, 2001.

\bibitem{navarro2016compact}
Gonzalo Navarro.
\newblock {\em Compact data structures: a practical approach}.
\newblock Cambridge University Press, 2016.

\bibitem{NStalg14}
Gonzalo Navarro and Kunihiko Sadakane.
\newblock Fully functional static and dynamic succinct trees.
\newblock {\em {ACM} Transactions on Algorithms}, 10(3):16:1--16:39, 2014.

\bibitem{pandey2017debgr}
Prashant Pandey, Michael~A Bender, Rob Johnson, and Rob Patro.
\newblock {deBGR}: an efficient and near-exact representation of the weighted
  de {Bruijn} graph.
\newblock {\em Bioinformatics}, 33(14):i133--i141, 2017.

\bibitem{philippe2013crac}
Nicolas Philippe, Mika{\"e}l Salson, Th{\'e}r{\`e}se Commes, and Eric Rivals.
\newblock {CRAC}: an integrated approach to the analysis of {RNA}-seq reads.
\newblock {\em Genome Biology}, 14(3):R30, 2013.

\bibitem{Raffinot2001}
Mathieu Raffinot.
\newblock On maximal repeats in strings.
\newblock {\em Information Processing Letters}, 80(3):165--169, 2001.

\bibitem{russo2009approximate}
Lu{\'\i}s Russo, Gonzalo Navarro, Arlindo~L Oliveira, and Pedro Morales.
\newblock Approximate string matching with compressed indexes.
\newblock {\em Algorithms}, 2(3):1105--1136, 2009.

\bibitem{SNsoda10}
K.~Sadakane and G.~Navarro.
\newblock Fully-functional succinct trees.
\newblock In {\em Proc. ACM-SIAM Symposium on Discrete Algorithms (SODA 2010)},
  pages 134--149, Austin, Texas, USA, 2010. ACM-SIAM.

\bibitem{schnattinger2012bidirectional}
Thomas Schnattinger, Enno Ohlebusch, and Simon Gog.
\newblock Bidirectional search in a string with wavelet trees and bidirectional
  matching statistics.
\newblock {\em Information and Computation}, 213:13--22, 2012.

\bibitem{SirenVMN08}
Jouni Sir{\'{e}}n, Niko V{\"{a}}lim{\"{a}}ki, Veli M{\"{a}}kinen, and Gonzalo
  Navarro.
\newblock Run-length compressed indexes are superior for highly repetitive
  sequence collections.
\newblock In {\em String Processing and Information Retrieval, 15th
  International Symposium, {SPIRE} 2008, Melbourne, Australia, November 10-12,
  2008.}, pages 164--175, 2008.

\bibitem{stoye2000affix}
Jens Stoye.
\newblock Affix trees.
\newblock Master's thesis, Universit\"{a}t Bielefeld, 2000.

\bibitem{strothmann2007affix}
Dirk Strothmann.
\newblock The affix array data structure and its applications to {RNA}
  secondary structure analysis.
\newblock {\em Theoretical Computer Science}, 389(1-2):278--294, 2007.

\bibitem{Wenger519025}
Aaron~M Wenger et~al.
\newblock Highly-accurate long-read sequencing improves variant detection and
  assembly of a human genome.
\newblock {\em bioRxiv preprint}, 2019.

\end{thebibliography}
\end{document}